\documentclass{llncs} 

\usepackage{hyperref}

\usepackage[T1]{fontenc}
\usepackage{textcomp}
\usepackage{listings}
\usepackage{graphicx}
\usepackage{amsmath} 
\usepackage{amsfonts} 
\usepackage{xspace}
\usepackage{relsize} 

\usepackage{scrextend} 

\usepackage[normalem]{ulem} 

\newtheorem{assumption}{Assumption}
\newtheorem{definitions}[definition]{Definitions}

\renewenvironment{proof}{\par\noindent\emph{Proof. }\ignorespaces}
 {\hspace*{\fill}{\scriptsize $\Box$}\medskip\par}

\newcommand{\intropara}[1]{\vspace*{-5pt}\paragraph{#1}}
\newcommand{\etal}{et al.\xspace}

\newcommand{\mycase}[1]{\mbox{{\underline{Case #1}}:\/}}

\newcommand{\braced}[1]{{ \left\{ #1 \right\} }}
\newcommand{\angled}[1]{{ \left\langle #1 \right\rangle }}

\newcommand{\FUNCTION}[2]{\ensuremath{\operatorname{\sf #1}({#2})}}

\newcommand{\queries}[1]{\Queries_{#1}}

\newcommand{\qinterval}{}
\newcommand{\Top}[1]{\FUNCTION{top}{#1}}

\newcommand{\opt}{\ensuremath{\operatorname{\sf opt}}}

\newcommand{\Reals}{\mathbb{R}}

\newcommand\eps\varepsilon

\newcommand{\textbfrm}[1]{{\text{\rm\bf{#1}}}}

\newcommand{\varQ}{v}  
\newcommand{\query}{\varQ}
\newcommand{\Keys}{{\cal K}}

\newcommand{\Queries}{{\cal Q}}
\newcommand{\Comparisons}{{\cal C}}

\newcommand{\Instance}{{\cal I}}

\newcommand{\compnode}[2]{\ensuremath{\angled{ \query #1 #2 }}}
\newcommand{\leafnode}[1]{\ensuremath{\angled{#1}}}
\newcommand{\weight}{{\omega}}

\newcommand{\TwCompTree}{{\mbox{\sc 2wcst}}\xspace}

\newcommand{\GBST}{{\sc gbst}\xspace}
\newcommand{\GBSTs}{{\sc gbst}s\xspace}
\newcommand{\EQK}[1]{e_{#1}} 
\newcommand{\SPLITK}[1]{s_{#1}} 

\newcommand{\NOvspace}[1]{}

\begin{document}

\title{Optimal search trees with 2-way comparisons\thanks{This is the full version of an extended abstract that appeared in ISAAC~\cite{ISAACPAPER}.} \NOvspace{-10pt}}%
\titlerunning{Optimal search trees}  
%

\author{
  Marek Chrobak\inst{1}\thanks{Research funded by NSF grants CCF-1217314 and CCF-1536026.}
  \and
  Mordecai Golin\inst{2}\thanks{Research funded by HKUST/RGC grant FSGRF14EG28.}
  \and
  J.~Ian Munro\inst{3}\thanks{Research funded by NSERC and the Canada Research Chairs Programme.}
  \and
 Neal E.~Young\inst{1}
}
\authorrunning{Chrobak et al.} 
%
%

\institute{University of California --- Riverside, Riverside, California, USA
  \and
  Hong Kong University of Science and Technology, Hong Kong, China
  \and
  University of Waterloo, Waterloo, Canada
	}

\maketitle 


\NOvspace{-0.2in}


\begin{abstract}
  In 1971, Knuth gave an $O(n^2)$-time algorithm 
  for the classic problem of finding an optimal binary search tree.
  Knuth's algorithm works only for search trees based on 3-way comparisons,
  but most modern computers support only 2-way comparisons
  ($<$, $\le$, $=$, $\ge$, and $>$).
  Until this paper, the problem of finding an optimal search tree using 2-way comparisons 
  remained open --- poly-time algorithms were known only for restricted variants. 
  We solve the general case, giving (i) an $O(n^4)$-time algorithm and
  (ii) an $O(n\log n)$-time additive-3 approximation algorithm.
  For finding optimal \emph{binary split trees},
  we \sout{(iii) obtain a linear speedup}
  and (iv) prove some previous work incorrect.
\end{abstract}


\NOvspace{-0.3in}\addtocounter{section}{-1}
\section{\emph{Erratum (March 2021)}}\label{sec: erratum}
{\em

  The proof of Theorem 3 here
  (and in the conference version of this paper~\cite{ISAACPAPER},
  and in versions 1--4 on arxiv.org) is incorrect.  
  For details of the error in the proof of Theorem 3, see the comments at the end of Section 3.1. 

  The remaining results are correct and have full proofs here,
  but some of the results have been updated substantially
  and now appear in updated manuscripts, as follows:
  \begin{itemize}
  \item For an update of Theorem 1, with a simpler algorithm for 2WCST, see~\cite{chrobak21simple}.

  \item For a result related to (but not supplanting) Theorem 2, see~\cite{chrobak21cost}.

  \item For an update of Theorem 4, with additional counterexamples for Spuler's dynamic program,
    see~\cite{chrobak21huang}.

  \end{itemize}
}


\NOvspace{-0.3in}
\section{Background and statement of results}\label{sec: introduction}

 
In 1971, Knuth~\cite{Knuth1971} gave an $O(n^2)$-time 
dynamic-programming algorithm for a classic problem: 
\emph{given a set $\Keys$ of keys and a probability distribution on queries,
find an optimal binary-search tree $T$.}
As shown in Fig.\,\ref{fig:knuth},
\begin{figure}[t]\centering
  \NOvspace{-2em}
  \includegraphics[height=0.9in]{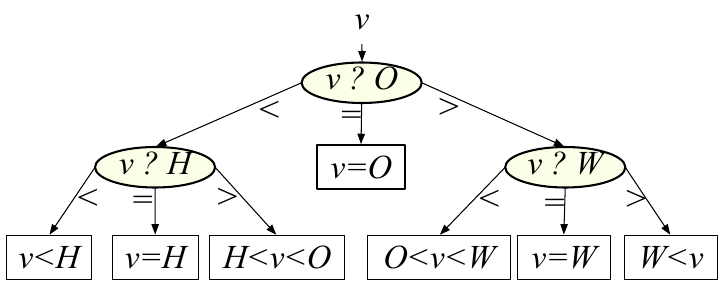}
  \NOvspace{-0.75em}
  \caption{A binary search tree $T$ using 3-way comparisons, for $\Keys=\{H,O,W\}$.}\label{fig:knuth}
  \NOvspace{-1.5em}
\end{figure}%
a search in such a tree for a given value $\query$
compares $\query$ to the root key,
then (i) recurses left if $\query$ is smaller,
(ii) stops if $\query$ equals the key,   
or (iii) recurses right if $\query$ is larger,
halting at a leaf.  
The comparisons made in the search must suffice 
to determine the relation of $\query$ to all keys in $\Keys$.
(Hence, $T$ must have $2|\Keys|+1$ leaves.)
$T$ is optimal if it has minimum \emph{cost},
defined as the expected number of comparisons assuming the query $\query$ 
is chosen randomly from the specified probability distribution.  
 
Knuth assumed \emph{three-way} comparisons at each node.
With the rise of higher-level programming languages,
most computers began supporting only two-way comparisons ($<,\le,=,\ge,>$).
In the 2nd edition of Volume 3 of \emph{The Art of Computer Programming}
\cite[\S6.2.2~ex.~33]{Knuth1998},
Knuth commented
\begin{quote}\footnotesize\NOvspace{-5pt}
  \emph{\ldots machines that cannot make three-way comparisons at once\ldots
    will have to make two comparisons\ldots 
    it may well be best to have a binary tree whose internal nodes specify either an equality test
    {\bf or} a less-than test but not both.  
  } 
\end{quote}
But Knuth gave no algorithm to find a tree built from \emph{two-way} comparisons 
(a \TwCompTree, as in Fig.~\ref{fig:2-way}(a)),
\begin{figure}\centering
  \NOvspace{-2em}
  \includegraphics[width=0.83\textwidth]{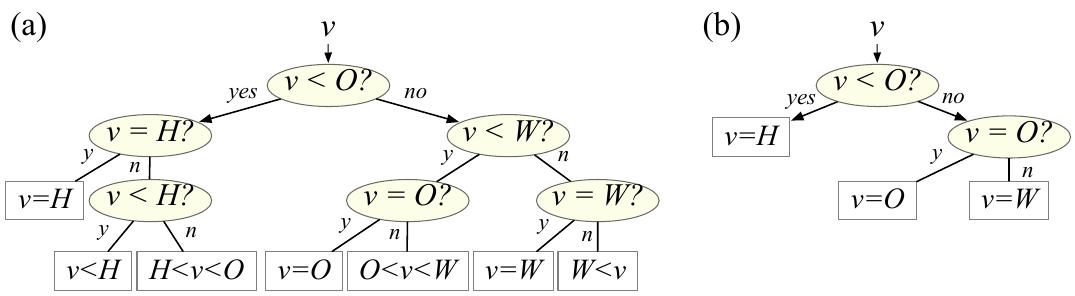}
  \caption{Two \TwCompTree{}s for $\Keys = \{H,O,W\}$;
    tree (b) only handles \emph{successful} queries.}\label{fig:2-way}
  \NOvspace{-1em}
\end{figure}
and, prior to the current paper, poly-time algorithms were known only for restricted variants.
Most notably, in 2002 Anderson \etal~\cite{Anderson2002} 
gave an $O(n^4)$-time algorithm for the \emph{successful-queries} variant of \TwCompTree,
in which each query $\query$ must be a key in $\Keys$,
so only $|\Keys|$ leaves are needed (Fig.~\ref{fig:2-way}(b)).
The \emph{standard} problem allows arbitrary queries,
so $2|\Keys|+1$ leaves are needed (Fig.~\ref{fig:2-way}(a)).
%
For the standard problem, no polynomial-time algorithm was previously known.
We give one for a more general problem that we call \TwCompTree:
\begin{theorem}\label{thm:general}
  \TwCompTree has an $O(n^4)$-time algorithm.
\end{theorem}
We specify an instance $\Instance$ of \TwCompTree
as a tuple $\Instance=(\Keys=\{K_1,\ldots,K_n\}, \Queries, \Comparisons, \alpha, \beta)$.
The set $\Comparisons$ of allowed comparison operators can be any subset of $\{<,\le,=,\ge,>\}$.  
The set $\Queries$ specifies the queries.
A solution is an optimal \TwCompTree $T$ among 
those using operators in $\Comparisons$ 
and handling all queries in $\Queries$.
This definition generalizes both standard \TwCompTree
(let $\Queries$ contain each key and a value between each pair of keys),
and the successful-queries variant (take $\Queries = \Keys$ and $\alpha \equiv 0$).
It further allows any query set $\Queries$ between these two extremes,
even allowing $\Keys \not\subseteq \Queries$.
As usual,
$\beta_i$ is the probability that $\query$ equals $K_i$;
$\alpha_i$ is the probability that $\query$ falls between keys $K_i$ and $K_{i+1}$
(except $\alpha_0=\Pr[\query < K_1]$ and $\alpha_n=\Pr[\query > K_n]$).
%
\footnote
{As defined here, a \TwCompTree $T$ must
determine the relation of the query $\query$ to every key in $\Keys$.  
More generally, one could specify any partition $\cal P$ of $\Queries$,
and only require $T$ to determine, if at all possible using keys in $\Keys$, which set $S\in\cal P$ contains $\query$.  
For example, 
if $\cal P = \{\Keys, \Queries\setminus\Keys\}$,
then $T$ would only need to determine whether $\query\in\Keys$.
We note without proof that Thm.~\ref{thm:general} extends to this more general formulation.}

To prove Thm.~\ref{thm:general},
we prove Spuler's 1994 ``maximum-likelihood'' conjecture:
\emph{in any optimal \TwCompTree tree, each equality comparison is to a key 
in $\Keys$ of maximum likelihood, given the comparisons so far}~\cite[\S6.4\,Conj.\,1]{Spuler1994}.
As Spuler observed, the conjecture implies an $O(n^5)$-time algorithm;
we reduce this to $O(n^4)$ using standard techniques and a new perturbation argument.
Anderson \etal proved the conjecture for their special case~\cite[Cor. 3]{Anderson2002}.
We were unable to extend their proof directly;
our proof uses a different local-exchange argument.

We also give a fast additive-3 approximation algorithm:
\begin{theorem}\label{thm:approximate}
  Given any instance $\Instance=(\Keys, \Queries, \Comparisons, \alpha, \beta)$ of \TwCompTree,
  one can compute a tree of cost at most the optimum plus 3,
  in $O(n\log n)$ time.
\end{theorem}
Comparable results were known for the successful-queries variant
$(\Queries=\Keys)$~\cite{Yeung1991,Anderson2002}.
We approximately reduce the general case to that case.

\intropara{Binary split trees}
``split'' each 3-way comparison in Knuth's 3-way-comparison model 
into two 2-way comparisons within the same node: 
an equality comparison (which, by definition, must be to the maximum-likelihood key)
and a ``$<$'' comparison (to any key)~\cite{Sheil1978,Comer1980,Huang1984,Perl1984,Hester1986}.
The fastest algorithms to find an optimal binary split tree take $O(n^5)$-time:
from 1984 for the successful-queries-only variant ($Q=K$)~\cite{Huang1984};
from 1986~for the standard problem
($Q$ contains queries in all possible relations to the keys in $K$)~\cite{Hester1986}.
\sout{We obtain a linear speedup:}~\footnote
{ERRATUM (MARCH 2021): SEE SECTIONS 0 AND 3.1.}
\begin{theorem}\label{thm:split}
  \sout{Given any instance $\Instance=(\Keys = \{K_1,\ldots,K_n\}, \alpha, \beta)$
  of the standard binary-split-tree problem,
  an optimal tree can be computed in $O(n^4)$ time.}$\,^5$
\end{theorem}
\sout{The proof uses our new perturbation argument  (Sec.~\ref{sec: non-distinct probabilities}) to reduce 
to the case when all $\beta_i$'s are distinct,
then applies a known algorithm~\cite{Hester1986}.}$\,^5$
The perturbation argument 
can \sout{also} be used to simplify Anderson \etal's algorithm~\cite{Anderson2002}.

\intropara{Generalized binary split trees} ({\sc gbst}s) are binary split trees 
without the maximum-likelihood constraint.
Huang and Wong~\cite{StephenHuang1984} (1984)
observe that relaxing this constraint allows cheaper trees
--- the maximum-likelihood conjecture fails here ---
and propose an algorithm to find optimal {\sc gbst}s.
We prove it incorrect!
\begin{theorem}\label{thm:flaw}
  Lemma~4 of~\cite{StephenHuang1984} is incorrect:
  there exists an instance
  --- a query distribution $\beta$ ---
  for which it does not hold,
  and on which their algorithm fails.
\end{theorem}
This flaw also invalidates two algorithms, proposed in Spuler's thesis~\cite{Spuler1994a},
that are based on Huang and Wong's algorithm.
We know of no poly-time algorithm to find optimal {\sc gbst}s. 
Of course, optimal \TwCompTree{}s are at least as good. 

\intropara{\TwCompTree without equality tests.}
Finding an optimal \emph{alphabetical encoding} has several poly-time algorithms:
by Gilbert and Moore --- $O(n^3)$ time, 1959~\cite{Gilbert1959};
by Hu and Tucker --- $O(n\log n)$ time, 1971~\cite{Hu1971};
and by Garsia and Wachs --- $O(n\log n)$ time but simpler, 1979~\cite{Garsia1977}.
The problem is equivalent 
to finding an optimal 3-way-comparison search tree
when the probability of querying any key is zero ($\beta \equiv 0$)~\cite[\S6.2.2]{Knuth1998}.
It is also equivalent to finding an optimal \TwCompTree
in the successful-queries variant with only ``$<$'' comparisons allowed
($\Comparisons=\{<\}, \Queries = \Keys$)~\cite[\S5.2]{Anderson2002}.
We generalize this observation to prove Thm.~\ref{thm:noequality}:
\begin{theorem}\label{thm:noequality}
  Any \TwCompTree instance
  $\Instance=(\Keys=\{K_1,\ldots,K_n\}, \Queries, \Comparisons, \alpha, \beta)$
  where $=$ is not in $\Comparisons$ (equality tests are not allowed),
  can be solved in $O(n\log n)$ time.
\end{theorem}

\begin{definitions}\label{def:preliminary}
  Fix an arbitrary instance $\Instance = (\Keys, \Queries, \Comparisons, \alpha, \beta)$. 

  For any node $N$ in any \TwCompTree $T$ for $\Instance$,
  $N$'s \emph{query subset}, $\queries N$, contains queries $\query\in\Queries$
  such that the search for $\query$ reaches $N$.
  The \emph{weight} $\weight(N)$ of $N$ is the probability that a random query $\query$
  (from distribution $(\alpha,\beta)$) is in $\queries N$.
  The weight $\weight(T')$ of any subtree $T'$ of $T$ is $\weight(N)$ where $N$ is the root of $T'$.

  Let $\compnode{<}{K_i}$ 
  denote an internal node having key $K_i$ and comparison operator $<$
  (define $\compnode{\le}{K_i}$ and  $\compnode{=}{K_i}$ similarly).
  Let $\leafnode{K_i}$ denote the leaf $N$ such that $\queries N = \{K_i\}$.
  Abusing notation, $\weight(K_i)$ is a synonym for $\weight(\leafnode{K_i})$, that is, $\beta_i$.

  Say $T$ is \emph{irreducible} if, 
  for every node $N$ with parent $N'$, $\queries N \ne \queries{N'}$. 
\end{definitions}

In the remainder of the paper, we assume that only comparisons
in $\{<,\le,=\}$ are allowed (i.e., $\Comparisons \subseteq \{<,\le,=\}$).
This is without loss of generality,
as ``$\query > K_i$'' 
and ``$\query \ge K_i$'' 
can be replaced, respectively, by ``$\query \le K_i$'' and ``$\query < K_i$.''





\NOvspace{-0.1in}
\section{Proof of Spuler's conjecture}\label{sec: central theorem}


Fix any irreducible, optimal \TwCompTree $T$ 
for any instance $\Instance = (\Keys, \Queries, \Comparisons, \alpha, \beta)$. 
\begin{theorem}[Spuler's conjecture]\label{thm:spuler}
  The key $K_a$ in any equality-comparison node $N = \compnode = {K_a}$
  is a maximum-likelihood key:
  $\beta_a = \max_i \{ \beta_i : K_i \in \queries N\}$.
\end{theorem}
The theorem will follow easily from Lemma~\ref{Lemma:Spuler}:

\begin{lemma}\label{Lemma:Spuler}
  Let internal node $\compnode = {K_a}$
  be the ancestor of internal node $\compnode = {K_z}$.
  Then $\weight(K_a) \ge \weight(K_z)$.
  That is, $\beta_a \ge \beta_z$.
\end{lemma}
\begin{proof}\emph{(Lemma~\ref{Lemma:Spuler})}
Throughout, ``$\compnode \prec {K_i}$'' denotes
a node in $T$ that does an inequality comparison ($\le$ or $<$, not $=$) to key $K_i$.
Abusing notation, in that context, ``$x \prec K_i$''
(or ``$x \not\prec K_i$'') denotes that $x$ passes (or fails) that comparison.

\begin{assumption}\label{ass:central }
{\textbfrm{(i)}}
All nodes on the path from $\compnode = {K_a}$ to $\compnode = {K_z}$ do inequality comparisons.
{\textbfrm{(ii)}}
Along the path, some other node $\compnode \prec {K_s}$
\emph{separates} key $K_a$ from $K_z$:
either $K_a \prec K_s$ but $K_z \not\prec K_s$,
or $K_z \prec K_s$ but $K_a \not\prec K_s$.
\end{assumption}
It suffices to prove the lemma assuming (i) and (ii) above.
(Indeed, if the lemma holds given (i), 
then, by transitivity, the lemma holds in general.
Given (i), if (ii) doesn't hold, then exchanging the two nodes
preserves correctness, changing the cost by $(\weight(K_a) - \weight(K_z))\times d$ for $d\ge 1$,
so $\weight(K_a) \ge \weight(K_z)$ and we are done.)

By Assumption~\ref{ass:central }, 
the subtree rooted at $\compnode = {K_a}$,
call it $T'$,
is as in Fig.\,\ref{fig:central 1}(a):\NOvspace{-5pt}

\begin{figure}[h]\centering\NOvspace{-10pt}
  \noindent\includegraphics[height=1in]{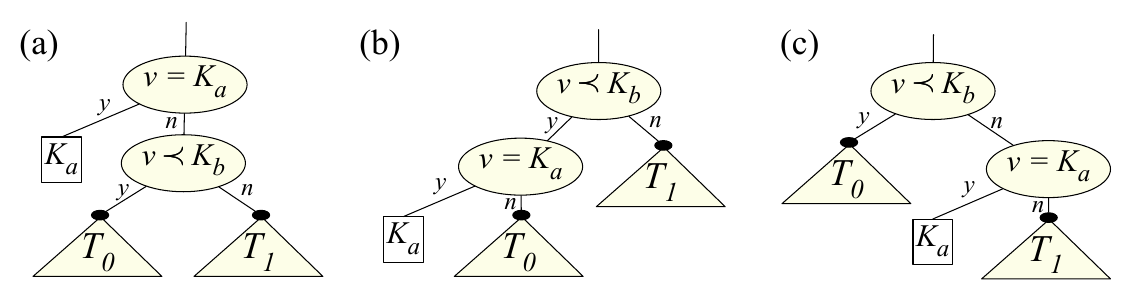}\NOvspace{-10pt}
  \caption{
    (a) The subtree $T'$ rooted at $\compnode = {K_a}$ and possible replacements (b), (c).
  }\label{fig:central 1}\NOvspace{-5pt}
\end{figure}

\NOvspace{-10pt}

\noindent
Let child $\compnode \prec {K_b}$, with subtrees $T_0$ and $T_1$, be as in Fig.~\ref{fig:central 1}.
\begin{lemma}\label{obs:central 1}
If $K_a\prec K_b$, then $\weight(K_a) \ge \weight(T_1)$,
else $\weight(K_a) \ge \weight(T_0)$.
\end{lemma}
\NOvspace{-5pt}
(This and subsequent lemmas in this section are proved in Appendix~\ref{sec:obs:central}.
The idea behind this one is that correctness is preserved by
replacing $T'$ by subtree (b) if $K_a\prec K_b$ or (c) otherwise,
implying the lemma by the optimality of $T$.)

\bigskip\noindent
\mycase{1} \emph{Child $\compnode \prec {K_b}$ separates $K_a$ from $K_z$.}
  If $K_a \prec K_b$, then $K_z \not\prec K_b$,
  so descendant $\compnode = {K_z}$ is in $T_1$,
  and, by this and Lemma~\ref{obs:central 1}, $\weight(K_a) \ge \weight(T_1) \ge \weight(K_z)$,
  and we're done.
  Otherwise $K_a \not\prec K_b$, so $K_z \prec K_b$,
  so descendant $\compnode = {K_z}$ is in $T_0$,
  and, by this and Lemma~\ref{obs:central 1}, $\weight(K_a) \ge \weight(T_0) \ge \weight(K_z)$,
  and we're done.

\smallskip\noindent
\mycase{2}
  \emph{Child $\compnode \prec {K_b}$ does not separate $K_a$ from $K_z$.}
  Assume also that descendant $\compnode = {K_z}$ is in $T_1$.
  (If descendant $\compnode = {K_z}$ is in $T_0$, the proof is symmetric,
  exchanging the roles of $T_0$ and $T_1$.)
  Since descendant $\compnode = {K_z}$ is in $T_1$, 
  and child $\compnode \prec {K_b}$ does not separate $K_a$ from $K_z$,
  we have $K_a \not\prec K_b$ and two facts:

\smallskip
\noindent\emph{Fact A:} $\weight(K_a) \ge \weight(T_0)$
    (by Lemma~\ref{obs:central 1}), and

\smallskip
\noindent\emph{Fact B:}
    the root of $T_1$ does an inequality comparison
    (by Assumption~\ref{ass:central }).

\smallskip

\noindent
By Fact B, subtree $T'$ rooted at $\compnode = {K_a}$ is as in Fig.\,\ref{fig:central 2}(a):
\begin{figure}[h]\centering\NOvspace{-10pt}
  \noindent\includegraphics[height=1.2in]{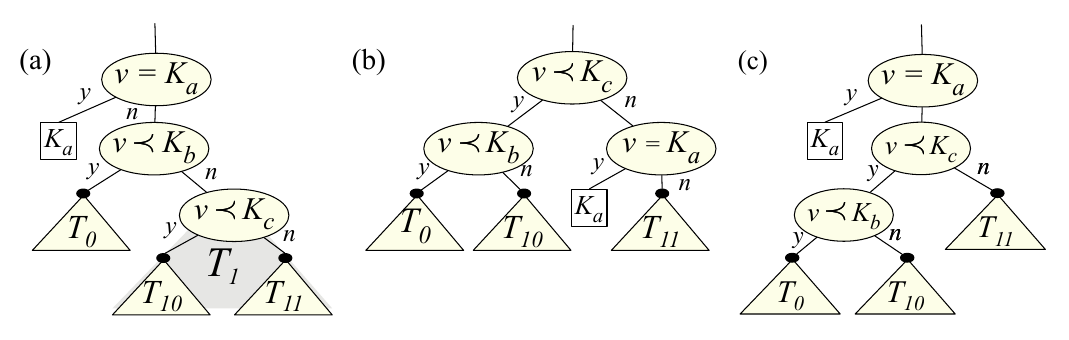}
  \NOvspace{-10pt}
  \caption{
    (a) The subtree $T'$ in Case 2, two possible replacements (b), (c).
  }\label{fig:central 2}\NOvspace{-10pt}
\end{figure}


\noindent
As in Fig.~\ref{fig:central 2}(a), let the root of $T_1$ be $\compnode \prec {K_c}$,
with subtrees $T_{10}$ and $T_{11}$.
\begin{lemma}\label{obs:central 2}\label{obs:central 3}
  (i) $\weight(T_0) \ge \weight(T_{11})$.
  (ii) If $K_a \not\prec K_c$, then $\weight(K_a) \ge \weight(T_1)$.
\end{lemma}
\NOvspace{-5pt}
 (As replacing $T'$ by (b) or (c) preserves correctness; proof in Appendix~\ref{sec:obs:central}.)
 \smallskip




\noindent
\mycase{2.1}
\emph{$K_a \not\prec K_c$.}
      By Lemma~\ref{obs:central 3}(ii), $\weight(K_a) \ge \weight(T_1)$.
      Descendant $\compnode = {K_z}$ is in $T_1$, so $\weight(T_1) \ge \weight(K_z)$.
      Transitively, $\weight(K_a) \ge \weight(K_z)$, and we are done.

\smallskip\noindent
\mycase{2.2}
\emph{$K_a \prec K_c$.}
      By Lemma~\ref{obs:central 2}(i), $\weight(T_0) \ge \weight(T_{11})$.
      By Fact A, $\weight(K_a) \ge \weight(T_{11})$. 

      If $\compnode = {K_z}$ is in $T_{11}$, 
      then $\weight(T_{11}) \ge \weight(K_z)$ and transitively we are done.

      In the remaining case, $\compnode = {K_z}$ is in $T_{10}$.
      $T$'s irreducibility implies $K_z \prec K_c$.
      Since $K_a \prec K_c$ also (Case 2.2), 
      grandchild $\compnode \prec {K_c}$ does not separate $K_a$ from $K_z$,
      and by Assumption~\ref{ass:central }
      the root of subtree $T_{10}$ does an inequality comparison.
      Hence, the subtree rooted at $\compnode  \prec {K_b}$ is 
      as in Fig.\,\ref{fig:central 4}(a):\NOvspace{-5pt}

      \begin{figure}[h]\centering\NOvspace{-10pt}
        \noindent\includegraphics[height=1.1in]{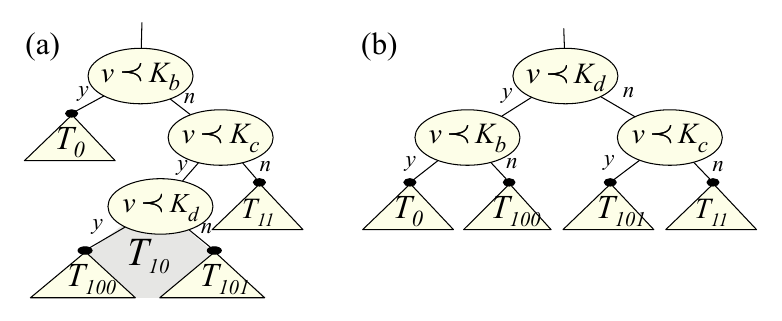}\NOvspace{-5pt}
        \caption{
          (a) The subtree rooted at $\compnode  \prec {K_b}$ in Case 2.2.
          (b) A possible replacement.}\label{fig:central 4}\NOvspace{-10pt}
      \end{figure}


      \NOvspace{-5pt}
      \begin{lemma}\label{obs:central 4}
        $\weight(T_0) \ge \weight(T_{10})$.
      \end{lemma}
      \NOvspace{-5pt}
      (Because replacing (a) by (b) preserves correctness; proof in Appendix~\ref{sec:obs:central}.)

      \smallskip

      Since descendant $\compnode = {K_z}$ is in $T_{10}$,
      Lemma~\ref{obs:central 4} implies $\weight(T_0) \ge \weight(T_{10}) \ge \weight(K_z)$.
      This and Fact A imply $\weight(K_a) \ge \weight(K_z)$.
      This proves Lemma~\ref{Lemma:Spuler}.
\end{proof}

\begin{proposition}\label{obs:canonical}
  If any leaf node $\leafnode {K_\ell}$'s parent $P$
  does not do an equality comparison against key $K_\ell$,
  then changing $P$ so that it does so
  gives an irreducible \TwCompTree $T'$ of the same cost.
\end{proposition}

\begin{proof}
  Since $\queries {\leafnode{K_\ell}} = \{K_\ell\}$ 
  and $P$'s comparison operator is in $\Comparisons\subseteq\{<,\le,=\}$,
  it must be that $K_\ell = \max\queries P$ or $K_\ell = \min\queries P$. 
  So changing $P$ to $\compnode = {K_\ell}$
  (with $\leafnode{K_\ell}$ as the``yes'' child
  and the other child the ``no'' child)
  maintains correctness, cost, and irreducibility.
\end{proof}

\begin{proof}\emph{(Thm.~\ref{thm:spuler})}
  Consider any equality-testing node $N = \compnode = {K_a}$
  and any key $K_z\in\queries N$.
  Since $K_z\in\queries N$, node $N$ has descendant leaf $\leafnode {K_z}$.
  Without loss of generality (by Proposition~\ref{obs:canonical}), 
  leaf  $\leafnode {K_z}$'s parent is $\compnode = {K_z}$.
  That parent is a descendant of $\compnode = {K_a}$,
  so $\weight(K_a) \ge \weight(K_z)$ by Lemma~\ref{Lemma:Spuler}.
\end{proof}




\NOvspace{-0.1in}
\section{Proofs of Thm.~\ref{thm:general} (algorithm for \TwCompTree) \sout{and Thm.~\ref{thm:split}}}%
\label{sec: algorithm}\label{sec:proof:general}%
\label{subsec:mainalg}

\newcommand{\SUBP}{_k^{\mathsmaller \prec}}

\newcommand{\subQ}{{\cal S}}

First we prove Thm.~\ref{thm:general}.
Fix an instance $\Instance=(\Keys,\Queries,\Comparisons,\alpha,\beta)$.
Assume for now that all probabilities in $\beta$ are distinct.
For any query subset $\subQ \subseteq \Queries$,
let $\opt(\subQ)$ denote the minimum cost of any \TwCompTree
that correctly determines all queries in subset $\subQ$
(using keys in $\Keys$, comparisons in $\Comparisons$,
and weights from the appropriate restriction of $\alpha$ and $\beta$ to $\subQ$).
Let $\weight(\subQ)$ be
the probability that a random query $\query$ is in $\subQ$.
The cost of any tree for $\subQ$ is the weight of the root ($= \weight(\subQ)$)
plus the cost of its two subtrees,
yielding the following dynamic-programming recurrence:
\begin{lemma}\label{obs:recurrence}
  For any query set $\subQ\subseteq \Queries$
  not handled by a single-node tree, 
  \[
  \opt(\subQ) \,=\,
  \weight(\subQ)+
  \min\begin{cases}
    \displaystyle \min_k\, \opt(\subQ\setminus \{k\})  \text{~~~(if ``$=$'' is in $\Comparisons$, else $\infty$)} \hfill ~~~ (i)
\\
    \displaystyle\min_{k,\prec} \,\opt(\subQ\SUBP) + \opt(\subQ\setminus\subQ\SUBP),
    \hfill ~~~(ii)
  \end{cases}
  \]
  where $k$ ranges over $\Keys$,
  and $\prec$ ranges over the allowed inequality operators (if any),
  and $\subQ\SUBP = \{\query\in \subQ : \query \prec k\}$.
\end{lemma}
Using the recurrence naively to compute $\opt(\Queries)$ yields
exponentially many query subsets $\subQ$, because of line (i).
But, by Thm.~\ref{thm:spuler},
we can restrict $k$ in line (i) to be the maximum-likelihood key in $\subQ$.
With this restriction,
the only subsets $\subQ$ that arise are intervals within $\Queries$,
minus some most-likely keys.
Formally, for each of $O(n^2)$ key pairs $\{k_1, k_2\}\subseteq\Keys\cup\{-\infty,\infty\}$
with $k_1<k_2$,
define four \emph{key intervals}
\[
\begin{array}{rcl@{~~~~}rcl}
  \qinterval(k_1,k_2)&=& \{\query\in\Queries : k_1 < \query < k_2\},
  & \qinterval[k_1,k_2] &=& \{\query\in\Queries : k_1 \le \query \le k_2\}, \\
  \qinterval(k_1,k_2] &=& \{\query\in\Queries : k_1 < \query \le k_2\}, 
  & \qinterval[k_1,k_2) &=&\{\query\in\Queries : k_1 \le \query < k_2\}.
\end{array}
\]
For each of these $O(n^2)$ key intervals $I$, and each integer $h\le n$,
define $\Top{I,h}$ to contain the $h$ keys in $I$ with the $h$ largest $\beta_i$'s.
Define $\subQ(I,h) = I\setminus \Top{I,h}$.  
Applying the restricted recurrence to $\subQ(I,h)$ gives a simpler recurrence:
\begin{lemma}\label{obs:recurrence3}
  If $\subQ(I,h)$ is not handled by a one-node tree,
  then $\opt(\subQ(I,h))$ equals
  \[
  \weight(\subQ(I,h))+
  \min\begin{cases}
    \displaystyle \opt(\subQ(I,h+1))  \text{~~~(if equality is in $\Comparisons$, else $\infty$)} \hfill ~~~ (i)
\\
    \displaystyle\min_{k,\prec}
    \,\opt(\subQ(I\SUBP,h\SUBP))
    + \opt(\subQ(I\setminus I\SUBP,h-h\SUBP)),
    \hfill ~~~(ii)
  \end{cases}
  \]
  where key interval $I\SUBP = \{\query\in I : \query \prec k\}$,
  and $h\SUBP = |\Top{I,h}\cap I\SUBP|$.
\end{lemma}

Now, to compute $\opt(\Queries)$, each query subset that arises
is of the form $\subQ(I,h)$ where $I$ is a key interval and $0\le h \le n$.
With care, each of these $O(n^3)$ subproblems can be solved in $O(n)$ time, 
giving an $O(n^4)$-time algorithm.
In particular, represent each key-interval $I$ by its two endpoints.
For each key-interval $I$ and integer $h\le n$, precompute
$\weight(\subQ(I,h))$, and $\Top{I,h}$, and the $h$'th largest key in $I$.
Given these $O(n^3)$ values (computed in $O(n^3\log n)$ time),
the recurrence for $\opt(\subQ(I,h))$ can be evaluated in $O(n)$ time.
In particular, for line (ii), one can enumerate all $O(n)$ pairs $(k,h\SUBP)$ 
in $O(n)$ time total, and, for each, 
compute $I\SUBP$ and $I\setminus I\SUBP$ in $O(1)$ time.
Each base case can be recognized and handled (by a cost-0 leaf) in $O(1)$ time,
giving total time $O(n^4)$.  This proves Thm.~\ref{thm:general}
when all probabilities in $\beta$ are distinct;
Sec.~\ref{sec: non-distinct probabilities} finishes the proof.



\subsection{Perturbation argument; proofs of Theorems~\ref{thm:general} \sout{and~\ref{thm:split}}\,$^7$}%
\label{sec: non-distinct probabilities}


Here we show that,
without loss of generality, in looking for an optimal search tree, 
one can assume that the key probabilities (the $\beta_i$'s) are all distinct.
Given any instance $\Instance=(\Keys,\Queries,\Comparisons,\alpha,\beta)$,
construct instance $\Instance'=(\Keys,\Queries,\Comparisons,\alpha,\beta')$,
where $\beta'_j = \beta_j + j\eps$ and $\eps$ is a positive infinitesimal
(or $\eps$ can be understood as a sufficiently small positive rational).
To compute (and compare) costs of trees with respect to $\Instance'$,
maintain the infinitesimal part of each value separately
and extend linear arithmetic component-wise in the natural way:
\begin{enumerate}
  \item Compute $z\times(x_1 + x_2\,\eps)$ as $(z x_1)+(z x_2)\eps$, where $z,x_1,x_2$ are any rationals,
  \item compute $(x_1 + \eps x_2) + (y_1+\eps y_2)$ as $(x_1+x_2) + (y_1+y_2)\eps$, 
  \item and say $x_1 + \eps x_2 <  y_1 + \eps y_2$ iff $x_1 < y_1$, or\,  $x_1=y_1 \,\wedge\, x_2<y_2$.
\end{enumerate}

\begin{lemma}\label{obs:distinct}
  In the instance $\Instance'$, all key probabilities $\beta'_i$ are distinct.
  If a tree $T$ is optimal w.r.t.~$\Instance'$, 
  then it is also optimal with respect to $\Instance$.
\end{lemma}
\begin{proof}
  Let $A$ be a tree that is optimal w.r.t.~$\Instance'$.
  Let $B$ be any other tree,
  and let the costs of $A$ and $B$ under $\Instance'$ be, respectively,
  $a_1 + a_2\eps$ and $b_1 + b_2\eps$. 
  Then their respective costs under $\Instance$ are $a_1$ and $b_1$.
  Since $A$ has minimum cost under $\Instance'$,
  $a_1 + a_2\eps \le b_1 + b_2\eps$. 
  That is, either $a_1 < b_1$, or $a_1 = b_1$ (and $a_2 \le b_2)$.
  Hence $a_1 \le b_1$: that is,
  $A$ costs no more than $B$ w.r.t.~$\Instance$.
  Hence $A$ is optimal w.r.t.~$\Instance$.
\end{proof}



Doing arithmetic this way increases running time by a constant factor.\footnote
{For an algorithm that works with linear (or $O(1)$-degree polynomial) functions of $\beta$.}
This completes the proof of Thm.~\ref{thm:general}.
The reduction can also be used to avoid the significant effort
that Anderson \etal~\cite{Anderson2002} devote
to non-distinct key probabilities.

For computing optimal \emph{binary split trees} for unrestricted queries,
the fastest known time is $O(n^5)$, due to~\cite{Hester1986}.
But~\cite{Hester1986} also gives an $O(n^4)$-time algorithm
for the case of distinct key probabilities.
\sout{With the above reduction, the latter algorithm gives $O(n^4)$ time for the general case,
  proving Thm.~\ref{thm:split}.}~\footnote
{ERRATUM (MARCH 2021): This reasoning is incorrect for binary split trees,
  because perturbing the key probabilities can change which key is the maximum-likelihood key,
  changing the space of valid search trees.}



\NOvspace{-0.1in}
\section{Proof of Thm.~\ref{thm:approximate} (additive-3 approximation algorithm)}%
\label{sec: approximations}

Fix any instance $\Instance=(\Keys,\Queries,\Comparisons,\alpha,\beta)$.
If $\Comparisons$ is $\{=\}$ then the optimal tree can be found in $O(n\log n)$ time,
so assume otherwise.  In particular, $<$ and/or $\le$ are in $\Comparisons$.
Assume that $<$ is in $\Comparisons$ (the other case is symmetric).

The entropy 
$H_{\Instance}= - \sum_i \beta_i \log_2 \beta_i - \sum_i \alpha_i \log_2 \alpha_i$
is  a lower bound on $\opt(\Instance)$.
For the case $\Keys = \Queries$ and $\Comparisons = \{<\}$,
Yeung's $O(n)$-time algorithm~\cite{Yeung1991} constructs a \TwCompTree 
that uses only $<$-comparisons whose cost is at most $H_{\Instance}+2-\beta_1-\beta_n$.
We reduce the general case to that one, adding roughly one extra comparison.


Construct 
$\Instance'=(\Keys'=\Keys,\,\Queries'=\Keys, \,\Comparisons'=\{<\},\,\alpha',\beta')$
where each $\alpha'_i=0$ and each $\beta' _i = \beta_i+ \alpha_i$
(except $\beta'_1 = \alpha_0 + \beta_1 + \alpha_1$).
Use Yeung's algorithm~\cite{Yeung1991} to construct tree $T'$ for $\Instance'$.
Tree $T'$ uses only the $<$ operator,
so any query $v\in\Queries$ that reaches a leaf
$\leafnode{K_i}$ in $T'$ must satisfy $K_i \le v < K_{i+1}$
(or $v < K_2$ if $i=1$).
To distinguish $K_i = v$ from $K_i < v < K_{i+1}$,
we need only add one additional comparison at each leaf
(except, if $i=1$, we need two).\footnote
{If it is possible to distinguish $v=K_i$ from $K_i < v < K_{i+1}$,
then $\Comparisons$ must have at least one operator other than $<$,
so we can add either $\compnode = {K_i}$ or $\compnode \le {K_i}$.}
By Yeung's guarantee, $T'$ costs at most $H_{\Instance'} + 2  -\beta'_1 -\beta'_n$. 
The modifications can be done so as to increase the cost by at most $1 + \alpha_0 + \alpha_1$,
so the final tree costs at most $H_{\Instance'} + 3$.
By standard properties of entropy, 
$H_{\Instance'} \le H_{\Instance} \le \opt(\Instance)$,
proving Thm.~\ref{thm:approximate}.



\NOvspace{-0.1in}
\section{Proof of Thm.~\ref{thm:flaw} (errors in work on binary split trees)}%
\label{sec: error}

\begin{figure}[t]\centering\NOvspace{-3pt}
  \noindent\includegraphics[trim=0 1 0 1, clip, width=0.9\textwidth]{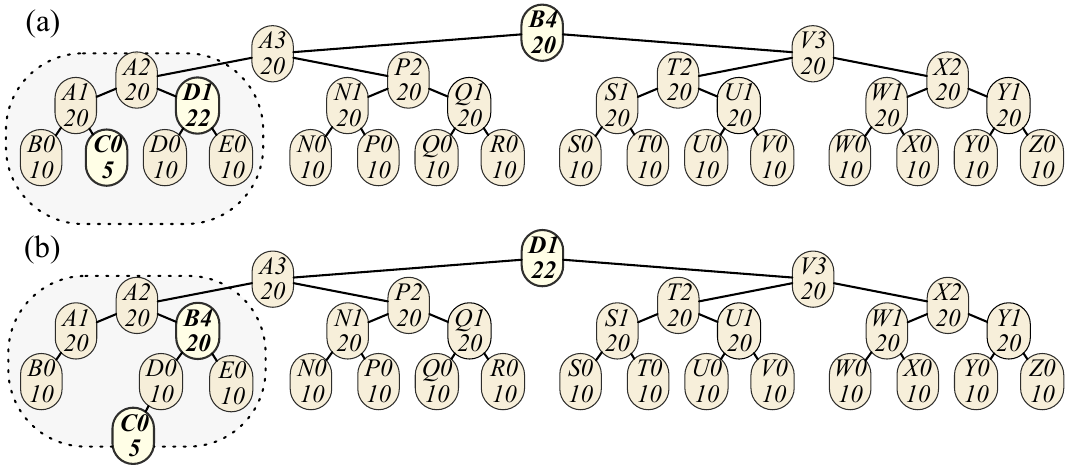}\NOvspace{-3pt}
  \caption{
    Two \GBSTs for an instance.
    Keys are ordered alphabetically ($A0 < A1 < A2 < A3 < B0  < \cdots$). 
    Each node shows its equality key and the frequency of that key;
    split keys are not shown.
    The algorithm of~\cite{StephenHuang1984} gives (a), of cost 1763, but (b) costs 1762.
  }\label{fig:error bad}\NOvspace{-0pt}
\end{figure}
A \emph{generalized binary split tree} (\GBST)
is a rooted binary tree where each node $N$
has an \emph{equality} key $\EQK N$ and a \emph{split} key $\SPLITK N$.
A search for query $\query\in\Queries$ starts at the root $r$.
If $\query = \EQK r$, the search halts.
Otherwise, the search recurses on the left subtree (if $\query < \SPLITK r$)
or the right subtree (if $\query\ge \SPLITK r$).
The \emph{cost} of the tree is the expected number of nodes (including, by convention, leaves)
visited for a random query $\query$.
Fig.~\ref{fig:error bad} shows two \GBSTs for a single instance.

To prove Thm.~\ref{thm:flaw}, we observe
that~\cite{StephenHuang1984}'s Lemma~4 and algorithm 
fail on the instance in Fig.~\ref{fig:error bad}.
There is a solution of cost only 1762  (in Fig.~\ref{fig:error bad}(b)),
but the algorithm gives cost 1763 for the instance  (as in Fig.~\ref{fig:error bad}(a)),
as can be verified by executing the Python code for the algorithm in Appendix~\ref{sec:code}.
The intuition is that the optimal substructure property fails for the subproblems 
defined by~\cite{StephenHuang1984}:
the circled subtree in (a) (with root $A2$) is cheaper
than the corresponding subtree in (b), but leads to larger global cost.
For more intuition and the full proof, see Appendix~\ref{app:flaw}.


\NOvspace{-0.1in}
\section{Proof of Thm.~\ref{thm:noequality} ($O(n\log n)$ time without equality)}%
\label{sec:nlogn}



Fix any \TwCompTree instance
$\Instance = (\Keys, \Queries, \Comparisons, \alpha, \beta)$ 
with $\Comparisons \subseteq\braced{ < , \le }$.
Let $n=|\Keys|$.
We show that,
in $O(n \log n)$ time,
one can compute an equivalent instance $\Instance' = (\Keys',\Queries',\Comparisons',\alpha',\beta')$
with $\Keys'=\Queries'$, $\Comparisons'=\{<\}$, and
$|\Keys'| \le 2n+1$.
(\emph{Equivalent} means that, given an optimal \TwCompTree $T'$ for $\Instance'$,
one can compute in $O(n\log n)$ time an optimal \TwCompTree $T$ for  $\Instance$.)
The idea is that, when $\Comparisons\subseteq\{<,\le\}$,
open intervals are functionally equivalent to keys.

Assume without loss of generality that $\Comparisons =\{<,\le\}$.
(Otherwise no correct tree exists unless $\Keys=\Queries$, and we are done.)
Assume without loss of generality that no two elements in $\Queries$
are equivalent (in that they relate to all keys in $\Keys$ in the same way;
otherwise, remove all but one query from each equivalence class).
Hence, at most one query lies between any two consecutive keys,
and $|\Queries|\le 2|\Keys|+1$.

Let instance $\Instance' = (\Keys',\Queries,\Comparisons',\alpha',\beta')$
be obtained by taking the key set $\Keys' =\Queries$ to be the key set,
but restricting comparisons to $\Comparisons'=\{<\}$
(and adjusting the probability distribution appropriately ---
take $\alpha' \equiv = 0$, take $\beta_i$ to be the probability associated with the $i$th query
--- the appropriate $\alpha_j$ or $\beta_j$).

Given any irreducible \TwCompTree $T$ for $\Instance$, 
one can construct a tree $T'$ for $\Instance'$ of the same cost as follows.
Replace each node $\compnode \le k$ with a node $\compnode < {q}$, 
where $q$ is the least query value larger than $k$
(there must be one, since $\compnode \le k$ is in $T$ and $T$ is irreducible).
Likewise, replace each node $\compnode < k$ with a node $\compnode < q$, 
where $q$ is the least query value greater than or equal to $k$
(there must be one, since $\compnode < k$ is in $T$ and $T$ is irreducible). 
$T'$ is correct because $T$ is.

Conversely, given any irreducible \TwCompTree $T'$ for $\Instance'$,
one can construct an equivalent \TwCompTree $T$ for $\Instance$ as follows.
Replace each node $N' = \compnode < q$ as follows.
If $q\in \Keys$, replace $N'$ by $\compnode < k$.
Otherwise, replace $N'$ by  $\compnode \le {k}$,
where key $k$ is the largest key less than $q$.
(There must be such a key $k$.  Node $\compnode < q$ is in $T'$ but $T'$ is irreducible,
so there is a query, and hence a key $k$, smaller than $q$.)
Since $T'$ correctly classifies each query in $\Queries$, so does $T$.

To finish, we note that the instance $\Instance'$ can be computed from $\Instance$
in $O(n\log n)$ time (by sorting the keys, under reasonable assumptions about $\Queries$),
and the second mapping (from $T'$ to $T$)  can be computed in $O(n\log n)$ time.
Since $\Instance'$ has $\Keys'=\Queries'$ and $\Comparisons= \{<\}$,
it is known~\cite{Knuth1971} to be equivalent to an instance of alphabetic encoding,
which can be solved in $O(n\log n)$ time~\cite{Hu1971,Garsia1977}.



\bibliographystyle{plainurl}
{\small
\bibliography{bib}
}

\newpage 
\section{Appendix}\label{sec: appendix}
\subsection{Python code for Thm.~\ref{thm:flaw} (\GBST algorithm of~\cite{StephenHuang1984})}%
\label{sec:code}
{\scriptsize
\begin{lstlisting} 
#!/usr/bin/env python3.4
import functools
memoize = functools.lru_cache(maxsize=None)

def huang1984(weights):
    "Returns cost as computed by Huang and Wong's GBST algorithm (1984)."

    n = len(weights)
    beta = {i+1 : weights[key] for i, key in enumerate(sorted(weights.keys()))}

    def is_legal(i, j, d): return 0 <= i <= j <= n and 0 <= d <= j - i

    @memoize
    def p_w_t(i, j, d):
        "Returns triple: (cost p[i,j,d], weight w[i,j,d], deleted keys for t[i,j,d])."

        interval = set(range(i+1, j+1))
        
        if d == j-i:                   # base case
            return (0, 0, interval)
        
        def candidates():              # Lemma 4 recurrence from Huang et al
            for k in interval:         #    k = index of split key
                for m in range(d+2):   #    m = num. deletions from left subtree
                    if is_legal(i, k-1, m) and is_legal(k-1, j, d-m+1):
                        cost_l, weight_l, deleted_l = p_w_t(i,  k-1,     m)
                        cost_r, weight_r, deleted_r = p_w_t(k-1,  j, d-m+1)
                        deleted = deleted_l .union( deleted_r )
                        x = min(deleted, key = lambda h : beta[h])
                        weight = beta[x] + weight_l + weight_r
                        cost = weight + cost_l + cost_r
                        yield cost, weight, deleted - set([x])

        return min(candidates())

    cost, weight, keys = p_w_t(0, n, 0)
    return cost

weights = dict(b4=20,
               a3=20, v3=20,
               a2=20, p2=20, t2=20, x2=20,
               a1=20, d1=22, n1=20, q1=20, s1=20, u1=20, w1=20, y1=20,
               b0=10, c0= 5, d0=10, e0=10, n0=10, p0=10, q0=10, r0=10,
               s0=10, t0=10, u0=10, v0=10, w0=10, x0=10, y0=10, z0=10)

assert huang1984(weights) == 1763  # Both assertions pass. The first is used in our Thm. 4.

weights['d1'] += 0.99              # Increasing a weight cannot decrease the optimal cost, but
assert huang1984(weights) <  1763  # in this case decreases the cost computed by the algorithm. 
\end{lstlisting}
}

\label{sec:appendix}

The extended abstract~\cite{ISAACPAPER} is essentially the body of this paper minus 
the remainder of this appendix.
The remainder of this appendix contains all proofs omitted from the extended abstract.






\subsection{Proof of Lemmas~\ref{obs:central 1}--\ref{obs:central 4} (in the proof of Spuler's conjecture)}%
\label{sec:obs:central}

We prove some slightly stronger lemmas that imply  Lemmas~\ref{obs:central 
1}--\ref{obs:central 4}.

Let $T$ be any irreducible, optimal \TwCompTree as in the proof of Lemma~\ref{Lemma:Spuler}.


\begin{lemma}[implies Lemma~\ref{obs:central 1}]\label{obs:app central 1}
Assume $T$ has a subtree as in Fig.\,\ref{fig:central 1}(a)
with nodes $\compnode = {K_a}$ and  $\compnode \prec {K_b}$.
{\textbfrm{(i)}} 
Replacing that subtree the one in Fig.\,\ref{fig:central 1}(b) (if $K_a\prec K_b$)
or the one in Fig.\,\ref{fig:central 1}(c) (if $K_a\not\prec K_b$)
preserves correctness.
{\textbfrm{(ii)}} If $K_a\prec K_b$, then $\weight(K_a) \ge \weight(T_1)$; 
otherwise $\weight(K_a) \ge \weight(T_0)$.
\end{lemma}
\begin{proof}
  Assume that $K_a\prec K_b$ (the other case is symmetric).
  By inspection of each case ($Q=K_a$ or $Q\ne K_a$),
  subtree (b) classifies each query $Q$ the same way subtree (a) does,
  so the modified tree is correct.  
  The modification changes the cost by $\weight(K_a) - \weight(T_1)$, 
  so (since $T$ has minimum cost) $\weight(K_a) \ge \weight(T_1)$.
\end{proof}

\begin{figure}[h]\centering\NOvspace{-5pt}
  \noindent\includegraphics[width=0.9\textwidth]{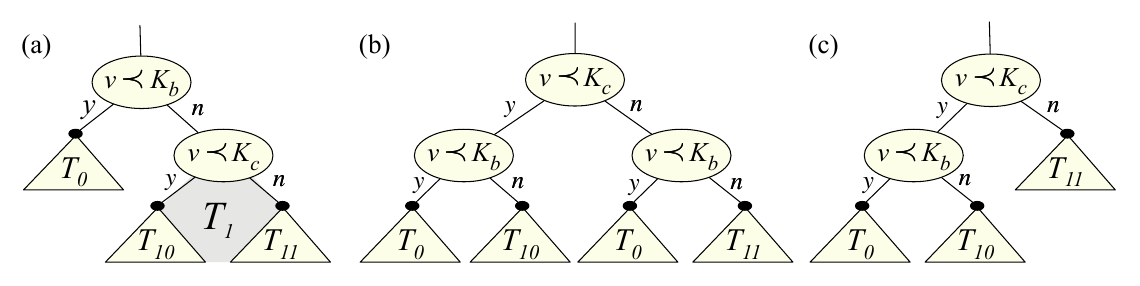}\NOvspace{-10pt}
  \caption{Lemma\,\ref{obs:app central 2} --- ``Rotating'' subtree (a) yields (c);
    the subtrees are interchangeable.}\label{fig:app central 2}\NOvspace{-10pt}
\end{figure}


\begin{lemma}[implies Lemma~\ref{obs:central 2}(i)]\label{obs:app central 2}
{\textbfrm{(i)}} If $T$ has either of the two subtrees in Fig.\,\ref{fig:app central 2}(a) or (c), 
then exchanging one for the other preserves correctness.
{\textbfrm{(ii)}} If $T$ has the subtree in Fig.\,\ref{fig:app central 2}(a), then
$\weight(T_0) \ge \weight(T_{11})$.
\end{lemma}
\begin{proof}\emph{Part (i).}
  The transformation from (a) to (c) is a standard rotation operation on
  binary search trees, but, since the comparison operators can be either $<$ or $\le$ in our context,
  we verify correctness carefully.

  By inspection, replacing subtree (a) by subtree (b) (in Fig.\,\ref{fig:app central 2})
  gives a tree that classifies all queries as $T$ does, and so is correct.

  Next we observe that, in subtree (b), replacing the right subtree
  by just $T_{11}$ (to obtain subtree (c)), maintains correctness.
  Indeed, since $T$ is irreducible,
  replacing (in (a)) the subtree $T_1$ by just $T_{11}$ would give an incorrect tree.
  Equivalently, $\exists Q.~ Q\not\prec K_b \,\wedge\, Q\prec K_c$.
  Equivalently,
  the right-bounded interval 
  $\{Q\in \Reals : Q\prec K_c\}$
  overlaps
  the left-bounded interval
  $\{Q\in \Reals : Q\not\prec K_b\}$.
  Equivalently,
  the complements of these intervals,
  namely
  $\{Q\in \Reals : Q\not\prec K_c\}$
  and
  $\{Q\in \Reals : Q\prec K_b\}$,
  are disjoint.
  Equivalently,
  $\forall Q.~Q\not\prec K_c \rightarrow Q \not\prec K_b$.
  Hence,
  replacing the right subtree of (b) by $T_{11}$ (yielding (c))
  maintains correctness.

  In sum, replacing subtree (a) by subtree (c) maintains correctness.
  This shows part (i).
  This replacement changes the cost by $\weight(T_0) - \weight(T_{11})$,
  so $\weight(T_0) \ge \weight(T_{11})$.  This proves part (iii).  
  The proof of (ii) is symmetric to the proof of (i).
\end{proof} 

\begin{figure}[h]\centering\NOvspace{0pt}
  \noindent\includegraphics[width=0.82\textwidth]{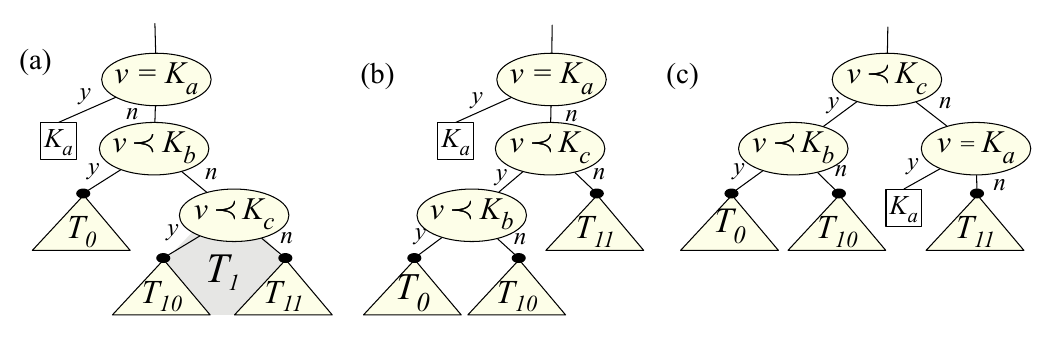}\NOvspace{-10pt}
  \caption{Lemma\,\ref{obs:app central 3} --- Subtrees (a) and (c) are interchangeable if $K_a \not\prec K_c$.}\label{fig:app central 3}\NOvspace{-5pt}
\end{figure}


\NOvspace{-5pt}
\begin{lemma}[implies Lemma~\ref{obs:central 3}(ii)]\label{obs:app central 3}
If $T$ has a subtree as in Fig.\,\ref{fig:app central 3}(a), and $K_a \not\prec K_c$, then
{\textbfrm{(i)}} replacing the subtree by Fig.\,\ref{fig:app central 3}(c) preserves correctness, and
{\textbfrm{(ii)}} $\weight(K_a) \ge \weight(T_1)$.
\end{lemma}
\begin{proof}
  (i) Assume $T$ has the subtree in  Fig.\,\ref{fig:app central 2}(a) (the other case is symmetric).
  By Lemma~\ref{obs:app central 2}(i) (applied to the subtree of (a) with root $\compnode \prec {K_b}$),
  replacing subtree (a) by subtree (b) gives a correct tree.
  Then, by Lemma~\ref{obs:app central 1}(i) 
  (applied to subtree (b), but note that node $\compnode \prec {K_c}$ in (b)
  takes the role of node $\compnode \prec {K_b}$ in Fig.\,\ref{fig:central 1}(a)!)
  replacing (b) by (c) gives a correct tree.
  This proves part (i).
  Part (ii) follows because 
  replacing (a) by (c) changes the cost of $T$ by $\weight(K_a) - \weight(T_1)$, 
  and $T$ has minimum cost, so $\weight(K_a) \ge \weight(T_1)$.
\end{proof}

\begin{figure}[h]\centering\NOvspace{-1em}
  \noindent\includegraphics[width=0.9\textwidth]{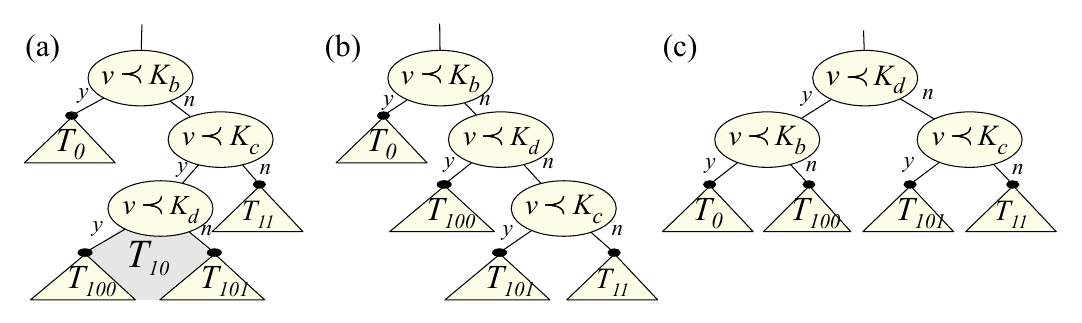}\NOvspace{-10pt}
  \caption{Lemma\,\ref{obs:app central 4} --- Subtrees (a) and (c) are interchangeable.}\label{fig:app central 4}\NOvspace{-10pt}
\end{figure}


\begin{lemma}[implies Lemma~\ref{obs:central 4}]\label{obs:app central 4}
  If $T$ has a subtree as in Fig.\,\ref{fig:app central 4}(a), then
  {\textbfrm{(i)}} replacing that subtree by the one in Fig.\,\ref{fig:app central 4}(c) preserves correctness, and
  {\textbfrm{(ii)}} $\weight(T_0) \ge \weight(T_{10})$.
\end{lemma}
\begin{proof}
  Applying Lemma~\ref{obs:app central 2}(i) to the subtree of (a) with root $\compnode \prec {K_c}$,
  replacing subtree (a) by subtree (b) gives a correct tree.\footnote
  {Technically, to apply Lemma~\ref{obs:app central 2}(i), we need (b) to be correct and \emph{irreducible}.
  The overall argument remains valid though as long as the tree from Fig.\,\ref{fig:app central 4}(a) is optimal.}
  Then, applying Lemma~\ref{obs:app central 2}(i) to the subtree of (b) with root $\compnode \prec {K_b}$,
  replacing subtree (b) by subtree (c) gives a correct tree.
  This shows part (i).
  Part (ii) follows, because replacing (a) by (c) changes the cost of $T$ by $\weight(T_0) - \weight(T_{10})$,
  so $\weight(T_0) \ge \weight(T_{10})$.
\end{proof}

\subsection{Proof of Thm.~\ref{thm:flaw} --- Huang and Wong's error }\label{app:flaw}%
%
\label{app: error}

A \emph{generalized binary split tree} (\GBST)
is a rooted binary tree where each node $N$
has an \emph{equality} key $\EQK N$ and a \emph{split} key $\SPLITK N$.
A search for query $\query\in\Queries$ starts at the root $r$.
If $\query = \EQK r$, the search halts.
Otherwise, the search recurses on the left subtree (if $\query < \SPLITK r$)
or the right subtree (if $\query\ge \SPLITK r$).
The \emph{cost} of the tree is the expected number of nodes (including, by convention, leaves)
visited for a random query $\query$.

\begin{figure}[h]\centering\NOvspace{-15pt}
  \noindent\includegraphics[height=0.9in]{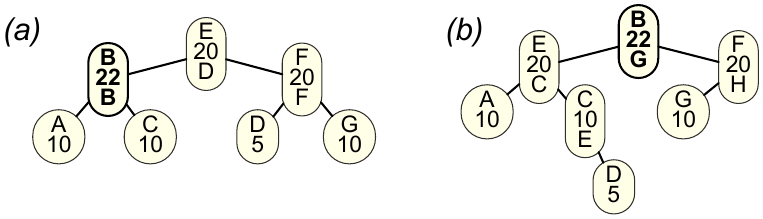}\NOvspace{-10pt}
  \caption{
    In a \GBST, each node has equality key, frequency, and (if internal) split key.
  }\label{fig:error huang}\NOvspace{-6pt}
\end{figure}


\newcommand{\KE}[1]{{\footnotesize\sf #1}}

Huang and Wong demonstrate that equality keys in optimal \GBSTs do not have the 
maximum-likelihood property~\cite{StephenHuang1984}.
Fig.~\ref{fig:error huang} shows their counterexample:
in the optimal \GBST (a), the root equality 
key is \KE{E} (frequency 20), not \KE{B} (frequency 22).
The cheapest tree with \KE{B} at the root is (b), and is more expensive.
Having \KE{B} at the root increases the cost because then the
other two high-frequency keys \KE{E} and \KE{F} have to be the children,
so the split key of the root has to split \KE{E} and \KE{F},
and low-frequency keys \KE{A}, \KE{C}, and \KE{D} all must be in the left subtree.

Following~\cite{StephenHuang1984}, restrict to successful queries ($\Keys=\Queries$).
Fix any instance $\Instance=(\Keys,\beta)$.
For any query interval $I=\{K_i,K_{i+1},\ldots,K_j\}$
and any subset $D\subseteq I$ of ``deleted'' keys,
let $\opt(I,D)$ denote the minimum cost of any \GBST
that handles the keys in $I\setminus D$.
This recurrence follows directly from the definition of \GBSTs:

\newcommand{\KEQ}{e}
\newcommand{\KLT}{s}

\begin{lemma}\label{obs:GBST recurrence}
  For any query set $I\setminus D$
  not handled by a single-node tree,  
  \[
  \opt(I,D) = 
  \weight(I\setminus D)+
  \displaystyle\min_{\KEQ, \KLT\in\Keys} 
  \,\opt(I_{{\mathsmaller <} \KLT}, \,D_{\KEQ} \cap I_{{\mathsmaller <} \KLT}) 
  + \opt(I_{{\mathsmaller \ge} \KLT},\, D_{\KEQ} \cap I_{{\mathsmaller \ge} \KLT})
  \]
  where $D_{\KEQ} = D \setminus \{\KEQ\}$,
  and $I_{{\mathsmaller <} \KLT} = \{\query\in I : \query < \KLT\}$
  and $I_{{\mathsmaller \ge} \KLT} = \{\query\in I : \query \ge \KLT\}$.
\end{lemma}
The goal is to compute $\opt(\Keys,\emptyset)$.
Using the recurrence above, exponentially many subsets $D$ arise.
This motivates the following lemma.
For any node $N$ in an optimal \GBST,
define $N$'s \emph{key interval}, $I_N$, and \emph{deleted-key set}, $D_N$,
according to the recurrence in the natural way.
Then the set $\queries N$ of queries reaching $N$ is $I_N\setminus D_N$,
and $D_N$ contains those keys in $I_N$ that are in equality tests above $N$,
and $I_N$ contains the key values that, if searched for in $T$ with the equality tests removed,
would reach node $N$.

\begin{lemma}[{\cite[Lemma~2]{StephenHuang1984}}]\label{obs:huang}
  For any node $N$ in an optimal \GBST,
  $N$'s equality key
  is a least-frequent key among those in $I_N$ that aren't equality keys
  in any of $N$'s subtrees:
  if $e_N = K_i$, then $\beta_i = \min \{\beta_j : K_j \in D_N\}$.
\end{lemma}
The proof is the same exchange argument
that shows our Assumption~\ref{ass:central }(ii).

{}\cite{StephenHuang1984} claims (incorrectly) that, by Lemma~\ref{obs:huang},
the desired value $\opt(\Queries,\emptyset)$
can be computed as follows.
For any key interval $I = \{K_{i+1},\ldots,K_j\}$ and $d\le n$, let
\begin{equation}\label{eq:huang p}
  p[i,j,d] = \min \{ \opt(I\setminus D) : D\subseteq I,\, |D|=d\}
\end{equation}
be the minimum cost of any \GBST for any query set $I\setminus D$
consisting of $I$ minus any $d$ deleted keys. 
Let $t[i,j,d]$ be a corresponding subtree of cost $p[i,j,d]$,
and let $w[i,j,d]$ be the weight of the root of $t[i,j,d]$.

Their algorithm uses the following (incorrect) recurrence (their Lemma~4):
\renewcommand{\KEQ}{\tilde e}
\begin{quote}\small\em
  \hspace*{1em}{$p[i,j,d] = \min ( w[i,j,d] + p[i,k-1,m] + p[k-1,j,d-m-1] )$}
  \smallskip 

  where the minimum is taken over all legal combinations of $k$'s and $m$'s
[and]
\smallskip 

\hspace*{1em}{$w[i,j,d] = w[i,k-1,m]+w[k-1,j,d-m-1] + \beta_x$}
\smallskip 

where $x$ is the index of the key of minimum frequency among those in range $\{K_{i+1},\ldots,K_j\}$
but outside $t[i,k-1,m]$ and $t[k-1,j,d-m+1]$\ldots''
\end{quote}

Next we describe their error.  
Recall that $p[i,j,d]$ chooses a subtree of \emph{minimum cost} (among trees with any $d$ keys deleted).
But this choice might not lead to minimum overall cost!
The reason is that the subtree's cost
does not suffice to determine the contribution to the overall cost:
the \emph{weight} of the subtree,
and the weights of the deleted keys and their eventual locations,
also matter.

\begin{figure}[h]\centering\NOvspace{-4pt}
  \noindent\includegraphics[trim=0 1 0 1, clip, height=1.05in]{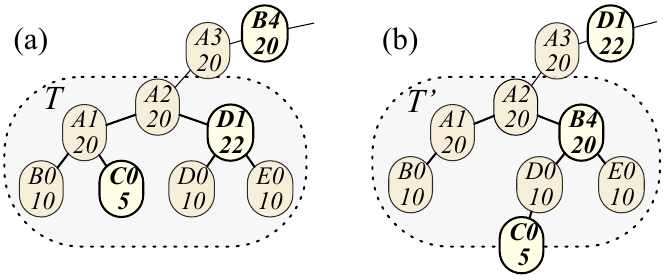}\NOvspace{-3pt}
  \caption{
    Trees $T$, $T'$ for 9-key interval $I$ with $d=2$. (Split keys not shown.)
  }\label{fig:error small}\NOvspace{-10pt}
\end{figure}

For example, consider $p[i,j,d]$ for the subproblem where $d=2$ 
and interval $I$ consists of the nine keys $I=\{A1,A2,A3,B0,B4,C0,D0,D1,E0\}$,
with the following weights:

\begin{quote}\centering\small

\begin{tabular}{|r|c|c|c|c|c|c|c|c|c|}\hline
	key:    & $A1$ &$A2$ & $A3$ & $B0$ & $B4$ & $C0$ & $D0$ & $D1$ & $E0$ \\ \hline
	weight: & 20 & 20 & 20 & 10 & 20 & 5 & 10 & 22 & 10 \\ \hline
\end{tabular}

\end{quote}

Fig.~\ref{fig:error small} shows two 7-node subtrees (circled and shaded), 
called $T$ and $T'$, involving these keys. These subtrees will be used in our counter-example, described below.
(The split key of each node is not shown in the diagram.)

Partition the set of possible trees $t[i,j,d]$ into two classes:
(i) those that contain $D1$ 
and (ii) those that don't (that is, $D1$ is a ``deleted'' key).
By a careful case analysis,\footnote
{In case (i), to minimize cost, the top two levels of the tree
must contain $D1$ and two other heavy keys from $\{A1,A2,A3,B4\}$.
$D1$ has to be the right child of the root,
because otherwise a key no larger than $B4$ is,
so the split key at the root has to be no larger than $B4$,
so the three light keys $\{C0,D0,E0\}$ 
have to all be in the right subtree,
so that one of them has to be at level four instead of level three,
increasing the cost by at least 5.
In case (ii), when $D1$ is deleted (not in the subtree),
by a similar argument, one of the light keys has to be at level four,
so $T'$ is best possible.}
subtree $T$ in Fig.~\ref{fig:error small}(a)
is a cheapest (although not unique) tree in class (i), 
while the 7-node subtree $T'$ in Fig.~\ref{fig:error small}(b)
is a cheapest tree in class (ii).
Further, the subtree $T'$ costs 1 more than the subtree $T$.
Hence, the algorithm of~\cite{StephenHuang1984} will choose
$T$, not $T'$, for this subproblem.

However, this choice is incorrect.
Consider not just the cost of tree, but also the effects of the choice
on the deleted keys' costs.
For definiteness, suppose the two deleted nodes become, respectively,
the parent and grandparent of the root of the subtree, 
as in (a) and (b) of the figure.
In (b), $C0$ is one level deeper than it is in (a), which increases the cost by 5,
but \emph{$D1$ is three levels higher}, which decreases the overall cost by $2\times 3 = 6$,
for a net decrease of 1 unit.
Hence, using $T$ instead of $T'$ ends up costing the overall solution 1 unit more.

This observation is the basis of the complete counterexample
shown in Fig.~\ref{fig:error bad}.
The counterexample extends the smaller example above 
by  appending two ``neutral'' subintervals, with 7 and 15 keys, respectively,
each of which (without any deletions) admits a self-contained balanced tree.
Keys are ordered alphabetically.
On this instance,
the algorithm of~\cite{StephenHuang1984} (and their Lemma~4) fail,
as they choose $T$ instead of $T'$ for the subproblem.
Fig.~\ref{fig:error bad}(a) shows the tree computed by their recurrence, of cost 1763.
(This can be verified by executing the Python code for the algorithm in Appendix~\ref{sec:code}.)
Fig.~\ref{fig:error bad}(b) shows a tree that costs 1 less.
This proves Thm.~\ref{thm:flaw}.


\paragraph{Spuler's thesis.}
In addition to Spuler's conjecture (and \TwCompTree algorithms that rely on his conjecture),
Spuler's thesis~\cite[\S6.4.1 Conj.\,1]{Spuler1994a}
also presents code for two additional algorithms that he claims, without proof, 
compute optimal \TwCompTree{}s independently of his conjecture.

First,~\cite[\S6.4.1]{Spuler1994a} 
gives code for the problem restricted to successful queries ($\Queries = \Keys$),
which it describes as a ``straightforward'' modification of Huang \etal's algorithm~\cite{StephenHuang1984} for generalized split trees
(an algorithm that we now know is incorrect, per our Thm.~\ref{thm:flaw}).
Correctness is addressed only by the remark that
\emph{``The changes to the optimal generalized binary split tree algorithm 
of Huang and Wong~\cite{StephenHuang1984}
to produce optimal generalized two-way comparison trees\,\footnote
{Spuler's \emph{generalized two-way comparison trees}
are exactly \TwCompTree{}s as defined herein.}%
are quite straight forward.''}

Secondly,~\cite[\S6.5]{Spuler1994a} 
gives code for the case of unrestricted queries,
which it describes as a ``not difficult'' modification 
of the preceding algorithm in \S6.4.1.
Correctness is addressed only by the following remark:
\emph{``The algorithm of the previous section assumes that $\alpha_i=0$ for all $i$.
However, the improvement of the algorithm to allow non-zero values of $\alpha$ is not difficult.''}
In contrast, Huang \etal explicitly mention that they were unable 
to generalize their algorithm to unrestricted queries~\cite{StephenHuang1984}.

Neither of Spuler's two proposed algorithms
is published in a peer-reviewed venue
(although they are referred to in~\cite{Spuler1994}).
They have no correctness proofs,
and are based on~\cite{StephenHuang1984},
which we now know is incorrect.
Given these considerations, 
we judge that Anderson \etal~\cite{Anderson2002}
give the first correct proof of a poly-time algorithm 
to find optimal \TwCompTree{}s
when only successful queries are allowed ($\Keys=\Queries$),
and that this paper gives the first correct proof for the general case.




\end{document}